\documentclass[a4paper,twocolumn,11pt]{article}
\usepackage[a4paper,margin=1in]{geometry}
\pdfoutput=1

\usepackage[utf8]{inputenc}
\usepackage[english]{babel}
\usepackage[T1]{fontenc}

\usepackage{amsmath,amssymb,mathtools}
\usepackage{bm}
\usepackage{braket}

\usepackage{graphicx}
\usepackage{xcolor}
\usepackage{booktabs}
\usepackage[expansion=false]{microtype}

\usepackage{algpseudocode}

\usepackage{amsthm}

\usepackage{xurl}

\usepackage{hyperref}
\usepackage[nameinlink,noabbrev]{cleveref}

\hypersetup{
  colorlinks=true,
  linkcolor=blue!55!black,
  citecolor=blue!55!black,
  urlcolor=blue!55!black,
  pdftitle={A Gauge-Covariant Theoretical Framework for Non-Abelian Holonomy Estimation and Feed-Forward Correction in Time-Bin Photonic Qudits},
  pdfauthor={N. Josef Bruzzese}
}

\newcommand{\repositoryurl}{https://github.com/njblottly/A-Gauge-Covariant-Theoretical-Framework-for-Non-Abelian-Holonomy-Estimation-and-Feed-Forward-Correct}
\newcommand{\zenododoi}{10.5281/zenodo.19945689}
\newcommand{\zenodourl}{https://doi.org/10.5281/zenodo.19945689}

\graphicspath{
  {./}
  {results/}
  {results/figs/}
  {../results/}
  {../results/figs/}
}
\DeclareGraphicsExtensions{.pdf,.PDF,.png,.PNG,.jpg,.JPG}

\numberwithin{equation}{section}

\newcommand{\dd}{\mathrm{d}}
\newcommand{\ii}{\mathrm{i}}

\DeclareMathOperator{\Tr}{Tr}

\DeclareMathOperator{\diag}{diag}
\DeclareMathOperator{\spec}{spec}
\DeclareMathOperator{\polar}{polar}
\DeclareMathOperator*{\argmin}{arg\,min}

\DeclarePairedDelimiter{\norm}{\lVert}{\rVert}

\newcommand{\cH}{\mathcal{H}}
\newcommand{\cS}{\mathcal{S}}
\newcommand{\cP}{\mathcal{P}}

\newenvironment{pseudofloat}[1][t]{%
  \begin{figure*}[#1]
  \hrule\vspace{0.6ex}
}{%
  \vspace{0.6ex}\hrule
  \end{figure*}
}

\usepackage{aliascnt}

\theoremstyle{plain}
\newtheorem{theorem}{Theorem}[section]

\newaliascnt{proposition}{theorem}
\newtheorem{proposition}[proposition]{Proposition}
\aliascntresetthe{proposition}

\newaliascnt{lemma}{theorem}
\newtheorem{lemma}[lemma]{Lemma}
\aliascntresetthe{lemma}

\theoremstyle{definition}
\newaliascnt{definition}{theorem}

\aliascntresetthe{definition}

\theoremstyle{remark}
\newaliascnt{remark}{theorem}
\newtheorem{remark}[remark]{Remark}
\aliascntresetthe{remark}

\crefname{theorem}{Theorem}{Theorems}
\Crefname{theorem}{Theorem}{Theorems}
\crefname{proposition}{Proposition}{Propositions}
\Crefname{proposition}{Proposition}{Propositions}
\crefname{lemma}{Lemma}{Lemmas}
\Crefname{lemma}{Lemma}{Lemmas}
\crefname{definition}{Definition}{Definitions}
\Crefname{definition}{Definition}{Definitions}
\crefname{remark}{Remark}{Remarks}
\Crefname{remark}{Remark}{Remarks}
\renewenvironment{abstract}
  {\par\noindent\bfseries\ignorespaces}
  {\par\medskip}

\title{A Gauge-Covariant Theoretical Framework for Non-Abelian Holonomy Estimation and Feed-Forward Correction in Time-Bin Photonic Qudits}
\author{N. Josef Bruzzese}

\begin{document}
\maketitle

\begin{abstract}
\bfseries
We develop a theoretical and computational framework for estimating and correcting non-Abelian geometric distortions in time-bin photonic qudit processing when the relevant encoded object is a transported logical subspace rather than a collection of independent rays. In such settings---for example under mode mixing, multiplexed routing, or effective degeneracies---the geometric contribution is naturally matrix-valued and is described by a Wilczek--Zee holonomy on a rank-$m$ subbundle of the ambient Hilbert space. The framework generalises prior Abelian time-bin Pancharatnam--Berry feed-forward calibration, in which geometric distortions are represented by bin-resolved scalar phases, to the non-Abelian, matrix-valued case. We construct a gauge-covariant discrete estimator from overlap matrices between successive subspace frames: the polar factor of each overlap gives a unitary backward frame comparator, and the adjoint comparators compose to approximate the forward path-ordered exponential of the Wilczek--Zee connection. We prove gauge covariance under frame changes, polar optimality of the local comparator, consistency under partition refinement, and perturbative stability under well-conditioned overlap errors. We then formulate left- and right-acting feed-forward correction rules for removing the estimated holonomy from an effective logical operation. The work does not assume a device-specific transfer matrix, loss model, detector model, or experimental calibration pipeline; numerical studies use synthetic non-Abelian transport models to validate covariance, convergence, conditioning dependence, and correction fidelity.
\end{abstract}

\section{Introduction}

Geometric phases are a foundational feature of quantum theory. The Abelian phase associated with cyclic adiabatic evolution was identified by Berry, while related formulations due to Pancharatnam, Aharonov--Anandan, and Samuel--Bhandari clarified its interferometric, nonadiabatic, and kinematic structure~\cite{Pancharatnam1956,Berry1984,AharonovAnandan1987,SamuelBhandari1988}. The corresponding non-Abelian generalisation was introduced by Wilczek and Zee for degenerate eigenspaces and was later developed in nonadiabatic settings by Anandan~\cite{WilczekZee1984,Anandan1988}. Matrix-valued holonomies have since become central to holonomic quantum computation and decoherence-free geometric control~\cite{Sjoqvist2012,Xu2012PRL}.

Time-bin encodings support high-dimensional photonic quantum information processing but remain sensitive to interferometric drift, mode mixing, and path-dependent calibration errors. Recent work has emphasised both the robustness and the practical challenges of time-bin photonic architectures, including high-dimensional time-bin entanglement, time-bin protocol stability, scalable time-bin photonic qudit logic, and the characterisation of linear optical devices~\cite{LaingOBrien2012,Yu2025NatCommun,White2025PRL,Singh2025TimeBins,Delteil2024QuditLogic}. In many time-bin calibration settings, the relevant geometric contribution is treated as a scalar phase, or as a list of scalar phases associated with independently transported time-bin components. This is appropriate when the encoded objects can be modelled as independent rays, or when the relevant geometric action is diagonal or commuting in the chosen logical basis.

The present paper studies what changes when that raywise or diagonal assumption is no longer the appropriate effective description. In architectures involving mode mixing, multiplexed routing, coupled interferometric stages, or approximately degenerate encoded sectors, the transported object may be better modelled as an entire logical subspace rather than as a collection of independent one-dimensional rays. In that case the geometric effect is no longer a scalar phase attached to a single ray. It is a unitary matrix acting on the encoded subspace.

This distinction is operationally important. A scalar geometric phase can be removed by estimating one number, or one independent phase per time bin. A non-Abelian holonomy cannot generally be removed in this way, because its matrix representation depends on the chosen basis inside the transported subspace and because geometric contributions accumulated at different points of the path need not commute. The calibration problem is therefore not simply phase estimation; it is the reconstruction of a unitary operator defined only up to a change of logical frame.

This perspective is especially natural for high-dimensional photonic encodings. In an idealised time-bin model, each logical basis state may be treated as an independent temporal mode. In a more general interferometric architecture, however, routing, multiplexing, mode mixing, or effective degeneracies can cause the relevant logical sector to evolve as a subspace. The basis used to describe that subspace is then partly conventional. A useful calibration procedure must therefore respect this gauge freedom rather than depend on an arbitrary choice of basis.

The appropriate continuum description is a Wilczek--Zee holonomy. If a rank-$m$ logical subspace $\cS(\lambda)\subset\cH$ is transported along a closed path $\gamma$ in parameter space, the geometric contribution is
\begin{equation}
  U_\gamma
  =
  \cP
  \exp\!\left(
    -\int_\gamma A
  \right),
\end{equation}
where $A=\Phi^\dagger\dd\Phi$ is the non-Abelian Berry connection associated with a local orthonormal frame $\Phi$ for the subspace. Since the choice of frame is arbitrary, physically meaningful quantities must be gauge-covariant or gauge-invariant.

The conceptual shift can be summarised as follows. In the Abelian picture, calibration estimates scalar phase errors and removes them by phase subtraction. In the non-Abelian subspace picture, calibration must estimate a matrix-valued holonomy, track how that matrix transforms under changes of logical frame, and apply its inverse on the correct side of an effective logical operation. The purpose of this paper is to formulate that non-Abelian replacement in a way that is mathematically gauge-covariant, numerically reproducible, and explicitly limited to the theoretical setting considered here.

\subsection{Relation to prior Abelian time-bin calibration}

This work is a non-Abelian generalisation of the Abelian time-bin geometric-phase calibration framework developed in Ref.~\cite{WeeBruzzese2026PB}. In that earlier setting, the relevant geometric distortion is diagonal in the time-bin basis: each bin acquires a scalar Pancharatnam--Berry phase, and feed-forward correction amounts to estimating and subtracting a bin-resolved phase profile. That framework is rooted in the Abelian geometric-phase structure introduced in Refs.~\cite{Pancharatnam1956,Berry1984,AharonovAnandan1987,SamuelBhandari1988}. The present work addresses the more general situation in which the transported object is not a collection of independent rays but an $m$-dimensional logical subspace. In that case the geometric contribution is no longer a diagonal phase transformation but a matrix-valued Wilczek--Zee holonomy~\cite{WilczekZee1984,Anandan1988}. Thus the scalar phase correction
\begin{equation}
  \phi_j \mapsto \phi_j-\widehat{\phi}_j
\end{equation}
is replaced by the non-Abelian correction
\begin{equation}
  V_{\mathrm{eff}}\mapsto \widehat U_\gamma^\dagger V_{\mathrm{eff}},
\end{equation}
or by the corresponding right-acting version depending on the circuit convention. The earlier Abelian protocol is therefore recovered when the transported logical sector decomposes into independent one-dimensional rays or when all relevant holonomies commute and are simultaneously diagonal in the chosen time-bin basis.

\subsection{Scope and validity assumptions}

\begin{quote}
\textbf{Scope.}
This paper is theoretical and computational. It studies unitary holonomy reconstruction for transported rank-$m$ logical subspaces from adjacent-frame overlap data. The overlap data may come from an analytic model, a numerical transfer-matrix simulation, prior device characterisation, or synthetic validation data. The present manuscript does not perform device tomography, does not model detector noise, and does not claim an experimental implementation.
\end{quote}

Loss and leakage are included only insofar as they leave a well-defined, well-transmitted $m$-dimensional sector from which orthonormal frames can be extracted. The reconstruction is reliable only when adjacent overlap matrices remain well-conditioned. If the minimum singular value of any adjacent overlap approaches zero, the local unitary comparison problem becomes ill-conditioned and the reconstructed holonomy should be treated as unreliable even though the polar product itself remains unitary.

\subsection{Main contributions}

The paper develops a gauge-covariant discrete holonomy estimator based on polar decomposition of adjacent overlap matrices, together with a sign and ordering convention distinguishing backward frame comparators from forward coefficient transport. The discrete estimator is shown to be gauge-covariant, with spectral and Wilson-loop diagnostics that are gauge-invariant, and the local comparator is identified as the nearest unitary to the raw overlap matrix via a polar-optimality result. A continuum consistency theorem establishes convergence under partition refinement, while a perturbative stability statement quantifies how reconstruction error depends on overlap conditioning. Building on the reconstruction, a feed-forward correction principle is derived for removing the estimated non-Abelian holonomy from an effective logical operation. The framework is supported by reproducible Mathematica/Wolfram Language validation of gauge covariance, Abelian reduction, convergence, noncommutativity, conditioning dependence, and correction fidelity.

\subsection{Positioning relative to existing methods}

The mathematical ingredients used here are closely related to established ideas in geometric phase theory, lattice Berry-curvature calculations, and adiabatic projector transport. Kato's formulation of adiabatic transport and Simon's geometric interpretation provide a standard projector-based background for subspace holonomy~\cite{Kato1950,Simon1983}. Gauge-covariant link variables and Wilson-loop-type discretisations are also widely used in numerical Berry-phase and Chern-number computations; the Fukui--Hatsugai--Suzuki construction is a canonical example of a gauge-invariant discretised Berry-connection method on a lattice~\cite{FukuiHatsugaiSuzuki2005}. Polar decompositions and their perturbation properties are standard tools in numerical linear algebra~\cite{Higham1986}.

The contribution of this paper is therefore not the claim that discrete non-Abelian holonomy or polar decomposition is new in isolation. Rather, the contribution is to assemble these ingredients into a gauge-covariant calibration and feed-forward framework tailored to transported time-bin photonic logical subspaces. The emphasis is on how adjacent-frame overlap data can be converted into a unitary holonomy estimate, how the resulting estimate transforms under logical-frame changes, how conditioning controls reliability, and how the reconstructed non-Abelian geometric factor should be removed from an effective logical operation.

This positioning also clarifies the role of the numerical validation. The tests below do not attempt to establish a new numerical method for Berry curvature in general. They verify that the proposed calibration pipeline has the properties needed for a later platform-specific implementation: covariance under frame changes, correct path ordering, stable polar unitarisation of overlaps, and successful left- or right-acting feed-forward correction.

\section{Notation and standing assumptions}

\subsection{Ambient space, logical subspaces, and frames}

Let $\cH$ denote the ambient Hilbert space, including the time-bin modes and any auxiliary or leakage modes used in an analytic model. The encoded logical system is assumed to occupy an $m$-dimensional subspace
\begin{equation}
  \cS(\lambda)\subset \cH,
  \qquad
  \dim \cS(\lambda)=m,
\end{equation}
depending on external control parameters $\lambda$. The orthogonal projector onto this subspace is denoted by
\begin{equation}
  P(\lambda):\cH\to\cS(\lambda).
\end{equation}

A local orthonormal frame for $\cS(\lambda)$ is represented by a matrix
\begin{equation}
  \Phi(\lambda)
  =
  \begin{pmatrix}
  | & & | \\
  u_1(\lambda) & \cdots & u_m(\lambda) \\
  | & & |
  \end{pmatrix},
\end{equation}
satisfying $\Phi(\lambda)^\dagger\Phi(\lambda)=I_m$, whose columns span $\cS(\lambda)$. The corresponding projector is
\begin{equation}
  P(\lambda)=\Phi(\lambda)\Phi(\lambda)^\dagger .
\end{equation}

The frame is not unique. Any smooth right multiplication
\begin{equation}
  \Phi(\lambda)\mapsto \Phi'(\lambda)=\Phi(\lambda)G(\lambda),
\end{equation}
with $G(\lambda)\in U(m)$, represents the same physical subspace. This is the gauge freedom of the encoded logical subspace.

This gauge freedom is analogous to choosing coordinates on a vector space. The subspace itself is physical, but the ordered list of basis vectors used to represent it is not. Consequently, matrix representatives of connections, holonomies, and gates can change under a frame transformation even when the underlying physical object has not changed. Throughout the paper, statements about matrices should therefore be read as statements in a chosen frame, while spectra and traces of closed-loop holonomies provide frame-independent diagnostics.

\subsection{Wilczek--Zee connection}

The Wilczek--Zee connection one-form associated with a chosen frame is
\begin{equation}
  A(\lambda):=\Phi(\lambda)^\dagger \dd\Phi(\lambda)
  \in \mathfrak u(m).
\end{equation}
Because $\Phi^\dagger\Phi=I_m$, differentiating gives
\begin{equation}
  A^\dagger=-A,
\end{equation}
so the connection is anti-Hermitian.

For a parameterised path $t\mapsto\lambda(t)$, $t\in[0,T]$, we write
\begin{equation}
  \Phi(t):=\Phi(\lambda(t)),
\end{equation}
and
\begin{equation}
  A_t(t):=\Phi(t)^\dagger\dot\Phi(t).
\end{equation}
The continuum forward holonomy is
\begin{equation}
  U_\gamma
  =
  \cP
  \exp\!\left(
    -\int_0^T A_t(t)\,\dd t
  \right),
  \label{eq:continuum_holonomy}
\end{equation}
where $\cP$ denotes path ordering.

\subsection{Discrete frames, overlaps, and endpoint identification}

Let
\begin{equation}
  0=t_0<t_1<\cdots<t_N=T
\end{equation}
be a discretisation of the path. Define
\begin{equation}
  \Phi_k:=\Phi(t_k),
  \qquad
  P_k:=P(t_k).
\end{equation}
The adjacent-frame overlap matrix is
\begin{equation}
  M_k:=\Phi_k^\dagger\Phi_{k+1},
  \qquad
  k=0,\ldots,N-1.
  \label{eq:overlap_matrix}
\end{equation}
This matrix compares the frame at step $k+1$ with the frame at step $k$.

For a physical projector loop one requires $P_N=P_0$, but the endpoint frame need not literally satisfy $\Phi_N=\Phi_0$. In general, one may have
\begin{equation}
  \Phi_N=\Phi_0 B,
  \qquad
  B\in U(m),
  \label{eq:endpoint_identification}
\end{equation}
where $B$ is the endpoint identification between the final frame and the chosen base frame. With the column-vector convention used below, the discrete product $\widehat U_\gamma$ maps initial coefficients in the frame $\Phi_0$ to final coefficients in the frame $\Phi_N$. Therefore the representative of the closed-loop holonomy in the base frame $\Phi_0$ is
\begin{equation}
  \widehat U_{\gamma}^{(0)}
  =
  B\,\widehat U_\gamma .
  \label{eq:base_frame_endpoint_holonomy}
\end{equation}
If the sampled frame loop is literally closed, then $B=I_m$ and $\widehat U_{\gamma}^{(0)}=\widehat U_\gamma$. If not, the endpoint identification in \cref{eq:base_frame_endpoint_holonomy} must be included before quoting eigenphases or Wilson-loop traces. This convention is necessary because such quantities are invariants of a closed transport represented in a single base frame.

A key convention used throughout the paper is the following. The polar unitary extracted from $M_k$ is a \emph{backward frame comparator}, while its adjoint is the corresponding \emph{forward coefficient transport}. Explicitly,
\begin{equation}
  W_k
  :=
  \polar(M_k)
  =
  M_k(M_k^\dagger M_k)^{-1/2},
  \label{eq:backward_comparator}
\end{equation}
and
\begin{equation}
  T_k:=W_k^\dagger .
  \label{eq:forward_transport}
\end{equation}
With column coefficient vectors, the forward discrete holonomy is ordered as
\begin{equation}
  \widehat U_\gamma
  :=
  T_{N-1}T_{N-2}\cdots T_1T_0.
  \label{eq:discrete_forward_holonomy}
\end{equation}
This convention is chosen so that
\begin{equation}
  \widehat U_\gamma
  \longrightarrow
  \cP\exp\!\left(
    -\int_0^T A_t(t)\,\dd t
  \right)
\end{equation}
under partition refinement.

\subsection{Standing assumptions}

Unless otherwise stated, we assume that each $\Phi_k$ has orthonormal columns and that each overlap matrix $M_k=\Phi_k^\dagger\Phi_{k+1}$ is invertible. We further assume that the overlap matrices are uniformly well-conditioned:
\begin{equation}
  \mu_{\min}
  :=
  \min_{0\le k\le N-1}\sigma_{\min}(M_k)
  \ge
  \mu>0 .
  \label{eq:standing_mu_min}
\end{equation}
For continuum convergence statements, the frame $\Phi(t)$ is assumed to be at least $C^2$ along the path, and closed-loop spectral diagnostics are evaluated after a fixed base-frame endpoint identification.

The invertibility and conditioning assumptions are part of the validity domain of the method, not merely technical conveniences. If $\mu_{\min}$ is small, adjacent frames have poor mutual overlap in at least one logical direction, the polar factor is ill-conditioned, and the reconstructed holonomy should not be interpreted as reliable without further refinement or additional physical information.

\section{Subspace transport and non-Abelian holonomy}

This section records the continuum geometric structure underlying the discrete estimator. The key point is that a transported logical subspace carries a natural gauge freedom: the physical subspace is specified by the projector $P=\Phi\Phi^\dagger$, while the choice of frame $\Phi$ inside the subspace is arbitrary up to a right action by $U(m)$.

The continuum structure used here is the standard Wilczek--Zee geometry of a transported degenerate subspace~\cite{WilczekZee1984}. The essential difference from the Abelian Berry phase~\cite{Berry1984} is that the connection takes values in the noncommuting Lie algebra $\mathfrak u(m)$ rather than in $\mathfrak u(1)$. This is why the holonomy requires path ordering and why its matrix representative is gauge-covariant rather than gauge-invariant. Related nonadiabatic generalisations and holonomic-control applications appear in Refs.~\cite{Anandan1988,Sjoqvist2012,Xu2012PRL}.

\subsection{Gauge transformation of the Wilczek--Zee connection}

\begin{lemma}[Gauge transformation law]
\label{lem:gauge_connection}
Under a frame change
\begin{equation}
  \Phi(\lambda)\mapsto \Phi'(\lambda)=\Phi(\lambda)G(\lambda),
  \qquad
  G(\lambda)\in U(m),
\end{equation}
the Wilczek--Zee connection
\begin{equation}
  A=\Phi^\dagger \dd\Phi
\end{equation}
transforms as
\begin{equation}
  A\mapsto A'
  =
  G^\dagger A G+G^\dagger\dd G.
  \label{eq:gauge_transformation_connection}
\end{equation}
\end{lemma}

\begin{proof}
Using $\Phi'=\Phi G$,
\begin{align}
  A'
  &=
  (\Phi G)^\dagger\dd(\Phi G) \\
  &=
  G^\dagger\Phi^\dagger(\dd\Phi)G
  +
  G^\dagger\Phi^\dagger\Phi\,\dd G \\
  &=
  G^\dagger A G+G^\dagger\dd G,
\end{align}
since $\Phi^\dagger\Phi=I_m$.
\end{proof}

The inhomogeneous term $G^\dagger\dd G$ shows that the connection itself is not gauge-invariant. This is expected: $A$ depends on the arbitrary frame choice. The physical information is instead encoded in the conjugacy class of the closed-loop holonomy.

\subsection{Holonomy and conjugacy invariants}

For a path $t\mapsto\lambda(t)$, $t\in[0,T]$, the forward coefficient transport associated with the connection is
\begin{equation}
  \dot U(t)=-A_t(t)U(t),
  \qquad
  U(0)=I_m.
  \label{eq:continuum_transport_ode}
\end{equation}
The solution is the path-ordered exponential in \cref{eq:continuum_holonomy}. For a closed loop $\gamma$, we write this endpoint transport as $U_\gamma$.

Under a gauge transformation $\Phi(t)\mapsto\Phi(t)G(t)$, the solution of \cref{eq:continuum_transport_ode} transforms as
\begin{equation}
  U(t)\mapsto U'(t)=G(t)^\dagger U(t)G(0).
  \label{eq:continuum_transport_gauge_transform}
\end{equation}
For a closed loop with $G(T)=G(0)$, and after using the same base-frame endpoint identification before and after the transformation, the holonomy transforms by conjugation:
\begin{equation}
  U_\gamma\mapsto U_\gamma'
  =
  G(0)^\dagger U_\gamma G(0).
  \label{eq:continuum_holonomy_conjugation}
\end{equation}
Therefore the spectrum of $U_\gamma$ and the Wilson-loop-type traces $\Tr(U_\gamma^r)$, for $r=1,2,\ldots$, are gauge-invariant.

\subsection{Projector formulation}

Although the frame-based connection is convenient for computations, the underlying geometry can also be expressed in terms of the projector path $P(t)$. The projector is gauge-invariant:
\begin{equation}
  P(t)=\Phi(t)\Phi(t)^\dagger.
\end{equation}
Differentiating $P^2=P$ gives
\begin{equation}
  P\dot P P=0,
\end{equation}
which expresses the fact that $\dot P$ only couples the subspace to its orthogonal complement at first order.

Parallel transport of a state vector $|\psi(t)\rangle\in\operatorname{Ran}P(t)$ may be characterised by
\begin{equation}
  P(t)|\psi(t)\rangle=|\psi(t)\rangle,
  \qquad
  P(t)|\dot\psi(t)\rangle=0.
  \label{eq:projector_parallel_transport}
\end{equation}
Writing $|\psi(t)\rangle=\Phi(t)c(t)$ gives
\begin{align}
  0
  &=
  P|\dot\psi\rangle \\
  &=
  \Phi\Phi^\dagger(\dot\Phi c+\Phi\dot c) \\
  &=
  \Phi(A_t c+\dot c).
\end{align}
Thus
\begin{equation}
  \dot c(t)=-A_t(t)c(t),
\end{equation}
which is the coefficient transport equation used above.

\subsection{Abelian reduction}

When $m=1$, the frame is a single normalised vector $|u(\lambda)\rangle$ defined up to a phase,
\begin{equation}
  |u(\lambda)\rangle\mapsto |u(\lambda)\rangle e^{\ii\chi(\lambda)}.
\end{equation}
The connection becomes the scalar one-form
\begin{equation}
  A=\braket{u|\dd u}.
\end{equation}
Under the phase change,
\begin{equation}
  A\mapsto A+\ii\,\dd\chi.
\end{equation}
For a closed loop, the holonomy is
\begin{equation}
  U_\gamma
  =
  \exp\!\left(
    -\oint_\gamma \braket{u|\dd u}
  \right).
\end{equation}
This is the usual Abelian geometric phase. Hence the non-Abelian framework reduces to standard Pancharatnam--Berry phase calibration when the transported logical subspace has rank one.

\section{Discrete holonomy estimation from overlap matrices}

We now construct the discrete holonomy estimator from adjacent-frame overlap matrices. The construction has three ingredients: the overlap matrix $M_k=\Phi_k^\dagger\Phi_{k+1}$, its polar unitary factor $W_k=\polar(M_k)$, and the forward transport $T_k=W_k^\dagger$. The adjoint in the final step is essential: $W_k$ compares the next frame to the current frame, whereas $T_k$ transports coefficient vectors forward.

The use of overlap matrices is motivated by the following elementary fact. If two adjacent frames span the same subspace and differ only by a unitary change of basis, then the overlap matrix between them is unitary. For genuinely moving subspaces, however, the adjacent frames do not span exactly the same subspace, and the overlap matrix is generally not unitary. Its singular values measure how well the two neighbouring subspaces overlap in each logical direction. The polar decomposition separates this information into a unitary part, which gives the best local basis comparator, and a positive part, which encodes the loss of perfect overlap between the neighbouring subspaces.

Thus the polar step should not be interpreted as an arbitrary numerical projection onto the unitary group. It is the canonical way to extract the unitary comparison contained in the raw overlap data. The positive polar factor is not discarded because it is unimportant; rather, it is converted into a conditioning diagnostic. When that positive factor has very small singular values, the local comparison problem itself becomes unstable.

\subsection{Local polar comparator}

For adjacent frames $\Phi_k$ and $\Phi_{k+1}$, define
\begin{equation}
  M_k:=\Phi_k^\dagger\Phi_{k+1}.
\end{equation}
If $M_k$ is nonsingular, its polar decomposition is
\begin{equation}
  M_k=W_kH_k,
\end{equation}
where $W_k\in U(m)$ and $H_k=(M_k^\dagger M_k)^{1/2}>0$. Equivalently,
\begin{equation}
  W_k
  =
  \polar(M_k)
  =
  M_k(M_k^\dagger M_k)^{-1/2}.
  \label{eq:polar_backward_comparator}
\end{equation}
We call $W_k$ the polar backward comparator. The corresponding forward coefficient transport is
\begin{equation}
  T_k:=W_k^\dagger.
  \label{eq:polar_forward_transport}
\end{equation}
The discrete forward holonomy is then
\begin{equation}
  \widehat U_\gamma
  :=
  T_{N-1}T_{N-2}\cdots T_1T_0.
  \label{eq:discrete_holonomy_estimator}
\end{equation}

\begin{remark}[Ordering convention]
The product in \cref{eq:discrete_holonomy_estimator} acts on column coefficient vectors. If $c_k$ is the coefficient vector in the frame $\Phi_k$, then $c_{k+1}=T_kc_k$, and hence after $N$ steps,
\begin{equation}
  c_N=
  T_{N-1}\cdots T_1T_0c_0.
\end{equation}
\end{remark}

\subsection{Polar optimality}

The polar comparator is not an arbitrary unitarisation of $M_k$. It is the nearest unitary matrix to $M_k$ in Frobenius norm.

This optimality property is useful conceptually because it explains why the polar unitary is the natural local comparator. Among all possible unitary matrices that could be assigned to a nonunitary overlap matrix, the polar factor is the one that changes the overlap least in Frobenius norm. In this sense, the algorithm extracts the minimally distorted unitary transport compatible with the adjacent-frame data.

\begin{proposition}[Polar optimality]
\label{prop:polar_optimality}
Let $M\in\mathbb C^{m\times m}$ be nonsingular and let
\begin{equation}
  W=\polar(M)=M(M^\dagger M)^{-1/2}.
\end{equation}
Then $W$ is the unique nearest unitary matrix to $M$ in Frobenius norm:
\begin{equation}
  W
  =
  \argmin_{Q\in U(m)}
  \norm{M-Q}_F .
\end{equation}
\end{proposition}

\begin{proof}
Let
\begin{equation}
  M=L\Sigma R^\dagger
\end{equation}
be a singular value decomposition, with $\Sigma=\diag(\sigma_1,\ldots,\sigma_m)$ positive because $M$ is nonsingular. The polar unitary is
\begin{equation}
  W=LR^\dagger.
\end{equation}
For any $Q\in U(m)$,
\begin{align}
  \norm{M-Q}_F^2
  &=
  \Tr[(M-Q)^\dagger(M-Q)] \\
  &=
  \Tr(M^\dagger M)+m
  -2\Re\Tr(Q^\dagger M).
\end{align}
Thus minimising $\norm{M-Q}_F^2$ is equivalent to maximising $\Re\Tr(Q^\dagger M)$. Define
\begin{equation}
  Y:=L^\dagger Q R.
\end{equation}
Since $Q,L,R$ are unitary, $Y\in U(m)$. Then
\begin{align}
  \Re\Tr(Q^\dagger M)
  &=
  \Re\Tr(Y^\dagger\Sigma) \\
  &=
  \sum_{a=1}^m \sigma_a \Re(Y_{aa}) \\
  &\le
  \sum_{a=1}^m \sigma_a,
\end{align}
with equality only when $Y=I_m$. Hence the unique maximiser is $Q=LR^\dagger=W$.
\end{proof}

\subsection{Gauge covariance of the discrete estimator}

Gauge covariance is the central consistency requirement for the discrete estimator. Since the frames $\Phi_k$ are arbitrary bases for the same physical subspaces, the reconstructed holonomy matrix cannot be expected to be identical after a gauge transformation. Equality of matrix entries would be the wrong criterion. The correct requirement is covariance: a change of base frame should conjugate the closed-loop holonomy. This ensures that basis-independent quantities, such as eigenvalues and traces of powers, are unchanged.

\begin{proposition}[Gauge covariance of the polar estimator]
\label{prop:gauge_covariance}
Let the discrete frames transform as
\begin{equation}
  \Phi_k\mapsto \Phi_k'=\Phi_kG_k,
  \qquad
  G_k\in U(m).
\end{equation}
Then the overlap matrices transform as
\begin{equation}
  M_k\mapsto M_k'
  =
  G_k^\dagger M_kG_{k+1}.
\end{equation}
The polar backward comparators and forward transports transform as
\begin{align}
  W_k &\mapsto W_k'=G_k^\dagger W_kG_{k+1}, \\
  T_k &\mapsto T_k'=G_{k+1}^\dagger T_kG_k.
\end{align}
For a closed loop with $G_N=G_0$, and after using the same endpoint base-frame identification, the discrete holonomy transforms by conjugation:
\begin{equation}
  \widehat U_\gamma
  \mapsto
  \widehat U_\gamma'
  =
  G_0^\dagger\widehat U_\gamma G_0.
\end{equation}
Consequently, $\spec(\widehat U_\gamma)$ and $\Tr(\widehat U_\gamma^r)$ for $r=1,2,\ldots$ are gauge-invariant.
\end{proposition}

\begin{proof}
The overlap transformation follows immediately:
\begin{align}
  M_k'
  &=
  (\Phi_kG_k)^\dagger(\Phi_{k+1}G_{k+1}) \\
  &=
  G_k^\dagger M_kG_{k+1}.
\end{align}
The polar decomposition is equivariant under left and right multiplication by unitary matrices. Indeed, if $M=W H$ is the polar decomposition of $M$, then
\begin{equation}
  G_k^\dagger M G_{k+1}
  =
  (G_k^\dagger W G_{k+1})
  (G_{k+1}^\dagger H G_{k+1}),
\end{equation}
where $G_k^\dagger W G_{k+1}$ is unitary and $G_{k+1}^\dagger H G_{k+1}$ is positive Hermitian. Therefore
\begin{equation}
  W_k'=G_k^\dagger W_kG_{k+1}.
\end{equation}
Taking adjoints gives
\begin{equation}
  T_k'=G_{k+1}^\dagger T_kG_k.
\end{equation}
Using the ordered product convention,
\begin{align}
  \widehat U_\gamma'
  &=
  T_{N-1}'T_{N-2}'\cdots T_0' \\
  &=
  G_N^\dagger
  (T_{N-1}T_{N-2}\cdots T_0)
  G_0.
\end{align}
For a closed loop with the same endpoint identification, $G_N=G_0$, so
\begin{equation}
  \widehat U_\gamma'
  =
  G_0^\dagger\widehat U_\gamma G_0.
\end{equation}
Spectra and traces of powers are invariant under conjugation.
\end{proof}

\subsection{Consistency under partition refinement}

\begin{theorem}[Consistency of forward polar transport]
\label{thm:consistency}
Let $\Phi(t)$ be a $C^2$ orthonormal frame along a path $t\in[0,T]$. Let
\begin{equation}
  0=t_0<t_1<\cdots<t_N=T
\end{equation}
be a partition with mesh size $h:=\max_k(t_{k+1}-t_k)$. Define
\begin{align}
  M_k &= \Phi(t_k)^\dagger\Phi(t_{k+1}), \\
  W_k &= \polar(M_k), \\
  T_k &= W_k^\dagger,
\end{align}
and $\widehat U_\gamma=T_{N-1}\cdots T_0$. Assume that all $M_k$ are nonsingular for sufficiently fine partitions. Then
\begin{equation}
  \widehat U_\gamma
  \longrightarrow
  \cP\exp\!\left(
    -\int_0^T A_t(t)\,\dd t
  \right),
\end{equation}
with $A_t(t)=\Phi(t)^\dagger\dot\Phi(t)$, as $h\to0$. For the first-order product rule above, the global error is $O(h)$.
\end{theorem}

\begin{proof}
Let $\Delta t_k=t_{k+1}-t_k$. Since $\Phi(t)$ is $C^2$,
\begin{equation}
  \Phi(t_{k+1})
  =
  \Phi(t_k)+\dot\Phi(t_k)\Delta t_k+O(\Delta t_k^2).
\end{equation}
Therefore, with $A_k:=A_t(t_k)$,
\begin{equation}
  M_k
  =
  I_m+A_k\Delta t_k+O(\Delta t_k^2).
\end{equation}
Because $\Phi^\dagger\Phi=I_m$, differentiating gives $A_k^\dagger=-A_k$. Hence
\begin{equation}
  M_k^\dagger M_k=I_m+O(\Delta t_k^2),
\end{equation}
and
\begin{equation}
  (M_k^\dagger M_k)^{-1/2}=I_m+O(\Delta t_k^2).
\end{equation}
Thus
\begin{equation}
  W_k=I_m+A_k\Delta t_k+O(\Delta t_k^2),
\end{equation}
and the forward step is
\begin{equation}
  T_k=W_k^\dagger
  =
  I_m-A_k\Delta t_k+O(\Delta t_k^2).
\end{equation}
The product $T_{N-1}\cdots T_0$ is therefore the standard first-order product approximation to the path-ordered exponential generated by $-A_t(t)$. The local error is $O(\Delta t_k^2)$, and summing over $O(h^{-1})$ steps gives a global $O(h)$ error.
\end{proof}

The theorem gives a conservative continuum guarantee. It assumes only smooth frame data and uses the first-order information contained in adjacent overlaps. In special synthetic benchmarks, higher apparent convergence order may be observed when midpoint exponentials or additional smoothness are used. Such behaviour is useful as an implementation check, but it should not be confused with the general theorem: the theorem states the robust condition under which the estimator approaches the Wilczek--Zee holonomy for arbitrary sufficiently smooth frame paths.

\subsection{Abelian discrete limit}

For $m=1$, the overlap matrix is a scalar,
\begin{equation}
  M_k=\braket{u_k|u_{k+1}}.
\end{equation}
The polar backward comparator is
\begin{equation}
  W_k=\frac{M_k}{|M_k|},
\end{equation}
and the forward step is
\begin{equation}
  T_k=W_k^\dagger=\frac{M_k^*}{|M_k|}.
\end{equation}
Therefore
\begin{equation}
  \widehat U_\gamma
  =
  \prod_{k=N-1}^{0}
  \frac{\braket{u_k|u_{k+1}}^*}
       {|\braket{u_k|u_{k+1}}|}.
\end{equation}
This is the discrete Abelian geometric phase convention corresponding to
\begin{equation}
  \exp\!\left(
    -\oint_\gamma \braket{u|\dd u}
  \right).
\end{equation}

\section{Minimal two-level non-Abelian example}

Before connecting the formalism to time-bin photonics, it is useful to record a minimal rank-two example that displays the essential non-Abelian features. Let the transported logical subspace have $m=2$, and suppose that in a chosen frame the connection along a parameter path is
\begin{equation}
  A_t(t)
  =
  \ii\left[
    a(t)\sigma_x+b(t)\sigma_z
  \right],
  \label{eq:minimal_pauli_connection}
\end{equation}
where $\sigma_x$ and $\sigma_z$ are Pauli matrices and $a(t),b(t)$ are real functions. The connection is anti-Hermitian, as required. If the ratio $a(t)/b(t)$ varies along the path, then generally
\begin{equation}
  [A_t(t),A_t(t')]
  =
  -2\ii\left[a(t)b(t')-b(t)a(t')\right]\sigma_y
  \ne 0 .
\end{equation}
The holonomy is therefore not determined by the ordinary exponential of the integrated connection; it requires path ordering:
\begin{equation}
  U_\gamma
  =
  \cP\exp\!\left[-\int_0^T A_t(t)\,\dd t\right].
\end{equation}

This example also illustrates why feed-forward correction has a left/right convention. If an intended logical gate $V$ does not commute with $U_\gamma$, then
\begin{equation}
  U_\gamma V \ne VU_\gamma .
\end{equation}
Thus a measured or modelled effective operation must first be assigned a convention,
\begin{equation}
  V_{\mathrm{eff}}\approx U_\gamma V
  \qquad\text{or}\qquad
  V_{\mathrm{eff}}\approx VU_\gamma ,
\end{equation}
before the inverse holonomy can be applied on the correct side. In the Abelian limit, where all relevant factors commute or reduce to scalar phases, this distinction disappears.

\section{A concrete time-bin photonic interpretation}

The formalism above is platform-independent, but it can be connected to a concrete time-bin photonic setting. A useful point of contact is the architecture proposed by Delteil, in which $N$ qubits are encoded in a single photon occupying a superposition of $2^N$ time bins and manipulated with a number of linear-optical elements scaling linearly in $N$~\cite{Delteil2024QuditLogic}. Consider more generally a device whose accessible modes consist of $d$ temporal bins, possibly augmented by ancillary routing or leakage modes. For each value of a control parameter $\lambda$---for example a vector of interferometer phases, switch settings, pulse-shaper phases, or electro-optic modulation amplitudes---the calibrated single-photon transformation may be represented by a transfer matrix
\begin{equation}
  \mathcal T(\lambda):\mathbb C^d\to\mathbb C^d.
\end{equation}
Such transfer information can in principle be obtained from linear-optical device characterisation or tomography~\cite{LaingOBrien2012}. The present work does not implement such a characterisation experimentally; it assumes that the required model or overlap data have already been supplied.

An $m$-dimensional logical time-bin code may be represented by a reference isometry
\begin{equation}
  \Phi_{\mathrm{in}}:\mathbb C^m\to\mathbb C^d.
\end{equation}
After propagation through the device, an effective transported logical subspace can be obtained by orthonormalising the columns of
\begin{equation}
  \mathcal T(\lambda)\Phi_{\mathrm{in}},
\end{equation}
or, more generally, by extracting an $m$-dimensional singular-vector subspace associated with the well-transmitted logical sector. This yields a frame
\begin{equation}
  \Phi(\lambda)
\end{equation}
for the transported logical subspace. Sampling a closed control path $\lambda_0,\lambda_1,\ldots,\lambda_N$ then gives adjacent overlaps
\begin{equation}
  M_k=\Phi(\lambda_k)^\dagger\Phi(\lambda_{k+1}),
\end{equation}
which are precisely the inputs to the polar holonomy estimator.

In this interpretation, the feed-forward correction is not necessarily a physical operation applied at every infinitesimal step. It may instead be compiled as a final logical correction acting on the encoded $m$-dimensional subspace. Depending on the hardware, this correction could be implemented by a programmable interferometric mesh, a pulse-shaping stage, or a calibrated logical-basis transformation. The present paper does not claim such an implementation. Its purpose is to specify the gauge-covariant mathematical object that such a correction would have to implement and to identify the conditioning diagnostics that determine whether the reconstruction is reliable.

This example also clarifies why the overlap formulation is useful. The algorithm does not require direct access to a continuum connection $A_t$. It requires only adjacent frame-comparison data, which can be generated by a model, simulation, or future device characterisation. This is the sense in which the proposed method is a calibration framework rather than an experimental demonstration.

\section{Stability and conditioning}

The preceding construction assumes exact overlap matrices. In practice, even in a theoretical numerical benchmark, one may have access only to approximate overlaps
\begin{equation}
  \widetilde M_k=M_k+E_k,
\end{equation}
where $E_k$ represents numerical error, model error, or characterisation error. The relevant question is how perturbations in $M_k$ propagate to perturbations in the polar unitary and, ultimately, to the reconstructed holonomy.

The key conditioning parameter is the smallest singular value of the overlap matrix,
\begin{equation}
  \sigma_{\min}(M_k).
\end{equation}
If $\sigma_{\min}(M_k)$ is bounded away from zero, the local polar comparator is stable. If $\sigma_{\min}(M_k)$ is small, then the local frame-comparison problem becomes ill-conditioned.

The conditioning issue has a simple geometric interpretation. If $\sigma_{\min}(M_k)$ is close to one, then every logical direction in $\cS(t_{k+1})$ has a strong projection onto $\cS(t_k)$, so the local comparison between frames is reliable. If $\sigma_{\min}(M_k)$ is close to zero, then at least one logical direction at the next step is almost orthogonal to the previous logical subspace. In that situation, there is not enough local overlap information to identify a stable unitary comparator. The instability is therefore not an artefact of the polar decomposition; it reflects a genuine loss of local distinguishability in the subspace-matching problem.

This is why the singular-value diagnostic is part of the algorithm rather than a secondary numerical detail. A holonomy estimate can be exactly unitary and nevertheless unreliable if it is assembled from poorly conditioned local overlaps. Reporting $\mu_{\min}$ alongside the holonomy spectrum makes this distinction explicit.

\subsection{Local polar stability}

\begin{proposition}[Local polar stability]
\label{prop:polar_stability}
Let $M\in\mathbb C^{m\times m}$ be nonsingular with $\sigma_{\min}(M)\ge \mu>0$, and let $\widetilde M=M+E$ with $\norm{E}_2\le \epsilon<\mu/2$. Then $\widetilde M$ is nonsingular and there exists a constant $C$ such that
\begin{equation}
  \norm{\polar(\widetilde M)-\polar(M)}_2
  \le
  C\frac{\epsilon}{\mu}.
  \label{eq:polar_stability_bound}
\end{equation}
The constant is not claimed to be sharp; it depends on the chosen norm convention and on the compact neighbourhood of nonsingular matrices on which the polar map is evaluated, but not on the particular perturbation $E$ within that neighbourhood.
\end{proposition}

\begin{proof}
By Weyl's singular value inequality,
\begin{equation}
  \sigma_{\min}(\widetilde M)
  \ge
  \sigma_{\min}(M)-\norm{E}_2
  >
  \frac{\mu}{2},
\end{equation}
so $\widetilde M$ is nonsingular. The polar map
\begin{equation}
  M\mapsto \polar(M)=M(M^\dagger M)^{-1/2}
\end{equation}
is smooth on the open set of nonsingular matrices. Standard perturbation results for the matrix polar decomposition imply local Lipschitz continuity on subsets with singular values bounded away from zero~\cite{Higham1986}. On the subset satisfying $\sigma_{\min}(M)\ge \mu/2$, its Fréchet derivative is bounded by a constant proportional to $\mu^{-1}$. Applying the mean-value theorem for matrix functions along the segment $M+sE$, $s\in[0,1]$, gives the stated estimate.
\end{proof}

The bound should be read as a local perturbative statement. For fixed overlap error $\epsilon$, the polar factor becomes more sensitive as the lower singular value bound $\mu$ decreases. Equivalently, the relevant small parameter is not just the absolute perturbation size but the dimensionless ratio $\epsilon/\mu$. This ratio is the quantity tested in the numerical conditioning study below.

\subsection{Accumulated holonomy perturbation}

\begin{proposition}[Accumulated holonomy stability]
\label{prop:accumulated_stability}
Let
\begin{equation}
  W_k=\polar(M_k),
  \qquad
  \widetilde W_k=\polar(\widetilde M_k),
\end{equation}
with $\norm{\widetilde W_k-W_k}_2\le \delta_k$. Let $T_k=W_k^\dagger$ and $\widetilde T_k=\widetilde W_k^\dagger$, and define
\begin{align}
  \widehat U_\gamma &= T_{N-1}\cdots T_0, \\
  \widetilde U_\gamma &= \widetilde T_{N-1}\cdots \widetilde T_0.
\end{align}
Then
\begin{equation}
  \norm{\widetilde U_\gamma-\widehat U_\gamma}_2
  \le
  \sum_{k=0}^{N-1}\delta_k.
  \label{eq:accumulated_stability_bound}
\end{equation}
In particular, if $\norm{\widetilde M_k-M_k}_2\le \epsilon_k$ and $\sigma_{\min}(M_k)\ge\mu>0$ for all $k$, then
\begin{equation}
  \norm{\widetilde U_\gamma-\widehat U_\gamma}_2
  \le
  C\sum_{k=0}^{N-1}\frac{\epsilon_k}{\mu}.
  \label{eq:accumulated_overlap_stability_bound}
\end{equation}
\end{proposition}

\begin{proof}
Use the telescoping identity
\begin{multline}
  \widetilde T_{N-1}\cdots\widetilde T_0
  -
  T_{N-1}\cdots T_0
  \\
  =
  \sum_{j=0}^{N-1}
  \widetilde T_{N-1}\cdots \widetilde T_{j+1}
  (\widetilde T_j-T_j)
  T_{j-1}\cdots T_0 .
\end{multline}
All factors outside $\widetilde T_j-T_j$ are unitary. Therefore
\begin{align}
  \norm{\widetilde U_\gamma-\widehat U_\gamma}_2
  &\le
  \sum_{j=0}^{N-1}
  \norm{\widetilde T_j-T_j}_2 \\
  &=
  \sum_{j=0}^{N-1}
  \norm{\widetilde W_j-W_j}_2.
\end{align}
The final claim follows by applying \cref{prop:polar_stability} to each overlap matrix.
\end{proof}

\subsection{Conditioning diagnostic}

The stability bounds identify a simple diagnostic for any theoretical or numerical implementation:
\begin{equation}
  \mu_{\min}
  :=
  \min_{0\le k\le N-1}
  \sigma_{\min}(M_k).
  \label{eq:min_singular_value_diagnostic}
\end{equation}
Reliable holonomy reconstruction requires $\mu_{\min}$ to remain bounded away from zero. In addition to reporting the reconstructed holonomy, one should therefore report conditioning data. A minimal gauge-compatible diagnostic set is given by $\spec(\widehat U_\gamma)$, $\Tr(\widehat U_\gamma^r)$, and $\mu_{\min}$.

\section{Feed-forward correction on the logical subspace}

The reconstructed holonomy can be used as a feed-forward correction when the effective logical operation contains a geometric contribution induced by subspace transport. The correction rule depends on whether the geometric contribution acts on the left or on the right of the intended logical gate.

This section separates two issues that are often conflated. The first is holonomy reconstruction: estimating the geometric unitary associated with the transported subspace. The second is correction convention: deciding on which side of the effective logical operation the inverse holonomy should be applied. In an Abelian problem the distinction is mostly invisible, because phases commute. In the present non-Abelian setting the distinction is essential.

\subsection{Left-acting geometric distortion}

Let $V\in U(m)$ denote the intended logical gate in the chosen base frame. Suppose that the effective logical operation is approximately decomposed as
\begin{equation}
  V_{\mathrm{eff}}
  \approx
  U_\gamma V,
  \label{eq:left_effective_decomposition}
\end{equation}
where $U_\gamma$ is the geometric holonomy associated with the transported logical subspace. Given the discrete estimate $\widehat U_\gamma$, define the corrected operation by
\begin{equation}
  V_{\mathrm{corr}}
  :=
  \widehat U_\gamma^\dagger V_{\mathrm{eff}}.
  \label{eq:left_feedforward_correction}
\end{equation}
If the decomposition \cref{eq:left_effective_decomposition} is exact and $\widehat U_\gamma=U_\gamma$, then $V_{\mathrm{corr}}=V$.

\subsection{Right-acting geometric distortion}

Some circuit or coefficient conventions may instead lead to
\begin{equation}
  V_{\mathrm{eff}}
  \approx
  VU_\gamma.
  \label{eq:right_effective_decomposition}
\end{equation}
In this case the correction is
\begin{equation}
  V_{\mathrm{corr}}
  :=
  V_{\mathrm{eff}}\widehat U_\gamma^\dagger.
  \label{eq:right_feedforward_correction}
\end{equation}
The distinction between \cref{eq:left_feedforward_correction,eq:right_feedforward_correction} is important in the non-Abelian case. When $m>1$, matrix multiplication is generally noncommutative, so left and right correction are not interchangeable.

In applications, the appropriate convention must be fixed by the circuit or coefficient ordering convention. For example, if the geometric transport is accumulated before the intended logical gate in the chosen representation, then the left-acting correction is appropriate. If instead the intended gate is followed by geometric transport, then the right-acting correction is the appropriate one. The formalism does not require one convention to be universal; it requires only that the estimated holonomy and the effective logical operation be represented in the same base frame. These correction rules apply only after a model-dependent identification of the geometric factor as left- or right-acting; the present work supplies the covariant reconstruction of that factor, not a universal hardware compilation procedure.

\subsection{Gauge covariance of the correction}

\begin{proposition}[Gauge-covariant feed-forward correction]
\label{prop:gauge_correction}
Assume the left-acting convention $V_{\mathrm{eff}}\approx U_\gamma V$. Suppose that, under a base-frame change $\Phi_0\mapsto \Phi_0G_0$ with $G_0\in U(m)$, the effective operation and reconstructed holonomy transform covariantly:
\begin{align}
  V_{\mathrm{eff}}
  &\mapsto
  G_0^\dagger V_{\mathrm{eff}}G_0, \\
  \widehat U_\gamma
  &\mapsto
  G_0^\dagger\widehat U_\gamma G_0.
\end{align}
Then the corrected operation $V_{\mathrm{corr}} = \widehat U_\gamma^\dagger V_{\mathrm{eff}}$ transforms covariantly:
\begin{equation}
  V_{\mathrm{corr}}
  \mapsto
  G_0^\dagger V_{\mathrm{corr}}G_0.
\end{equation}
\end{proposition}

\begin{proof}
Using the assumed transformation laws,
\begin{align}
  V_{\mathrm{corr}}'
  &=
  (\widehat U_\gamma')^\dagger V_{\mathrm{eff}}' \\
  &=
  (G_0^\dagger\widehat U_\gamma G_0)^\dagger
  (G_0^\dagger V_{\mathrm{eff}}G_0) \\
  &=
  G_0^\dagger
  \widehat U_\gamma^\dagger
  V_{\mathrm{eff}}
  G_0 \\
  &=
  G_0^\dagger V_{\mathrm{corr}}G_0.
\end{align}
\end{proof}

\subsection{Correction error}

In the left-acting convention, suppose
\begin{equation}
  V_{\mathrm{eff}}=U_\gamma V+R,
  \label{eq:effective_with_residual}
\end{equation}
where $R$ represents residual non-geometric implementation error. Then
\begin{align}
  V_{\mathrm{corr}}-V
  &=
  \widehat U_\gamma^\dagger V_{\mathrm{eff}}-V \\
  &=
  \left(\widehat U_\gamma^\dagger U_\gamma-I_m\right)V
  +
  \widehat U_\gamma^\dagger R .
\end{align}
Taking the operator norm and using unitarity gives
\begin{equation}
  \norm{V_{\mathrm{corr}}-V}_2
  \le
  \norm{\widehat U_\gamma-U_\gamma}_2
  +
  \norm{R}_2.
  \label{eq:correction_error_bound}
\end{equation}

A convenient gauge-invariant diagnostic for the corrected logical operation is the normalised unitary overlap
\begin{equation}
  F(V_{\mathrm{corr}},V)
  :=
  \frac{1}{m^2}
  \left|
    \Tr(V_{\mathrm{corr}}^\dagger V)
  \right|^2.
  \label{eq:gate_fidelity_diagnostic}
\end{equation}

\section{Protocol-level algorithm}

The algorithm takes as input either a sequence of frames $\{\Phi_k\}_{k=0}^N$ or a sequence of adjacent overlap matrices $\{M_k\}_{k=0}^{N-1}$. It outputs a reconstructed holonomy $\widehat U_\gamma$ and, when an effective logical operation $V_{\mathrm{eff}}$ is supplied, a corrected operation $V_{\mathrm{corr}}$.

\begin{pseudofloat}
\caption{Gauge-covariant holonomy reconstruction and feed-forward correction.}
\label{fig:pseudocode}
\begin{algorithmic}[1]
  \Require Frames $\{\Phi_k\}_{k=0}^{N}$ or overlaps $\{M_k\}_{k=0}^{N-1}$; optional effective gate $V_{\mathrm{eff}}$
  \Ensure Holonomy estimate $\widehat U_\gamma$; optional corrected gate $V_{\mathrm{corr}}$
  \For{$k=0$ to $N-1$}
    \If{frames are supplied}
      \State $M_k \gets \Phi_k^\dagger\Phi_{k+1}$
    \EndIf
    \State $W_k \gets M_k(M_k^\dagger M_k)^{-1/2}$ \Comment{polar backward comparator}
    \State $T_k \gets W_k^\dagger$ \Comment{forward coefficient transport}
  \EndFor
  \State $\widehat U_\gamma \gets T_{N-1}T_{N-2}\cdots T_1T_0$
  \If{$V_{\mathrm{eff}}$ is supplied}
    \State $V_{\mathrm{corr}} \gets \widehat U_\gamma^\dagger V_{\mathrm{eff}}$ \Comment{left-acting correction}
  \EndIf
  \State \Return $\widehat U_\gamma$ and, if computed, $V_{\mathrm{corr}}$
\end{algorithmic}
\end{pseudofloat}

For a right-acting geometric distortion, the final correction line should instead be replaced by
\begin{equation}
  V_{\mathrm{corr}}
  \gets
  V_{\mathrm{eff}}\widehat U_\gamma^\dagger.
\end{equation}

The dominant operation at each step is the polar decomposition of an $m\times m$ matrix. Using a singular value decomposition or eigendecomposition of $M_k^\dagger M_k$, this costs $O(m^3)$ per step, so the full reconstruction cost scales as $O(Nm^3)$.

\section{Reproducible numerical validation}

The framework above is theoretical, but it is supported by reproducible numerical validation. The repository accompanying this manuscript contains six Mathematica/Wolfram Language notebooks: gauge covariance, non-Abelian convergence, frame-based overlap reconstruction, feed-forward correction, conditioning/noise sensitivity, and an aggregate validation summary. The generated data are synthetic validation outputs, not experimental measurements.

The validation studies are designed to support the formal claims of the paper rather than to reproduce a specific device. In an experimental photonic setting, one could obtain the required overlap or transfer information from device characterisation or tomography methods such as linear-optical process reconstruction~\cite{LaingOBrien2012}. The present manuscript instead uses analytic and synthetic data so that the gauge covariance, convergence, and conditioning behaviour can be isolated from device-specific noise sources.

The relevance of such calibration tools is strengthened by recent progress in high-dimensional and robust time-bin photonic protocols, including single-photon time-bin qudit logic proposals~\cite{Yu2025NatCommun,White2025PRL,Singh2025TimeBins,Delteil2024QuditLogic}.

The purpose of these numerical studies is not to model a particular photonic chip. Instead, the tests isolate the mathematical claims of the paper. Each benchmark is designed to check one structural property: covariance under frame changes, convergence to a known continuum holonomy, correct reconstruction from overlap data, feed-forward removal of the reconstructed holonomy, or sensitivity to overlap conditioning. This makes the numerical results reproducible without requiring device-specific parameters.

The validation suite also distinguishes between two kinds of evidence. Connection-based tests check the non-Abelian ordered-product convention directly. Frame-based tests check the full pipeline from sampled frames to overlaps, polar comparators, and holonomy reconstruction. Noise and conditioning tests then probe the regime in which the assumptions of the stability theorem begin to fail.

\subsection{Validation overview}

The numerical suite tests the closed-loop gauge covariance of the estimator; the convergence of a discrete ordered product to a continuum non-Abelian holonomy; the convergence of the full frame-to-overlap-to-polar-holonomy pipeline against an analytic frame-loop reference; feed-forward correction for left- and right-acting geometric distortions; and stability scaling with the dimensionless perturbation ratio $\eta/\mu_{\min}$. The aggregate validation status was \texttt{PASS}: all checks for gauge covariance, unitarity, convergence, feed-forward correction, and scaled-noise behaviour returned true in the validation summary.

\begin{table*}[t]
\centering
\caption{Structure of the validation suite. The tests are synthetic checks of the theoretical pipeline rather than measurements from a photonic device.}
\label{tab:validation_structure}
\begin{tabular}{lll}
\toprule
Validation module & Quantity checked & Expected behaviour \\
\midrule
Gauge covariance & $\norm{\widehat U'_\gamma-G_0^\dagger \widehat U_\gamma G_0}$ & Numerical-zero residual \\
Abelian reduction & Scalar phase convention & Agreement with the Abelian phase limit \\
Non-Abelian convergence & $\norm{\widehat U_\gamma-U_{\mathrm{ref}}}$ & Decrease under refinement \\
Frame-overlap pipeline & Error from sampled frames alone & Agreement with analytic frame-loop reference \\
Feed-forward correction & $F(V_{\mathrm{corr}},V)$ & Near-unit fidelity when assumptions hold \\
Conditioning/noise & Error versus $\eta/\mu_{\min}$ & Growth as conditioning worsens \\
\bottomrule
\end{tabular}
\end{table*}

\subsection{Gauge covariance}

The gauge covariance test samples a synthetic frame loop and applies random closed-loop gauge transformations $G_N=G_0$. The residual
\begin{equation}
  \epsilon_{\mathrm{gauge}}
  =
  \norm{
    \widehat U_\gamma'
    -
    G_0^\dagger\widehat U_\gamma G_0
  }_F
\end{equation}
was at numerical precision. Across the tested refinements, the maximum gauge covariance residual was
\begin{equation}
  9.54\times 10^{-15},
\end{equation}
and the maximum gauge-run unitarity residual was
\begin{equation}
  1.45\times 10^{-14}.
\end{equation}

These residuals are at the level expected from floating-point arithmetic. The test confirms that changing the frame sequence alters the holonomy matrix only by the predicted base-frame conjugation. Thus the matrix entries themselves are not invariant, but the spectrum and Wilson-loop traces are stable under the gauge transformation.

\subsection{Non-Abelian connection benchmark}

The first convergence benchmark uses the anti-Hermitian connection
\begin{multline}
  A(t)
  =
  \ii\bigl[
    0.7\cos(t)\sigma_x
    \\
    +
    0.4\sin(2t)\sigma_y
    +
    0.2\sigma_z
  \bigr],
  \label{eq:numerical_nonabelian_connection}
\end{multline}
for $t\in[0,2\pi]$. For generic $t,t'$, $[A(t),A(t')]\ne0$, so path ordering is essential. The continuum reference is obtained from
\begin{equation}
  \dot U(t)=-A(t)U(t),
  \qquad
  U(0)=I_2.
\end{equation}
The reference eigenphases were $\{-0.70134,\;0.70134\}$. The Wilson traces for $r=1,2,3$ were approximately
\begin{equation}
  \{1.52795,\;0.33464,\;-1.01663\}.
\end{equation}

The Frobenius error decreased from $1.57\times 10^{-2}$ at the coarsest tested partition to $9.01\times 10^{-7}$ at the finest partition. A log--log fit gave an estimated convergence order
\begin{equation}
  2.00849.
\end{equation}

The observed order is higher than the conservative first-order consistency guarantee because this benchmark uses smooth midpoint exponentials for the connection-based ordered product. The purpose of the test is therefore not to replace the general theorem, but to verify that the implemented ordering and sign convention reproduce the expected non-Abelian continuum transport with rapid convergence in a smooth benchmark.

\subsection{Frame-based overlap pipeline}

The second convergence benchmark validates the full frame-to-overlap-to-polar-holonomy pipeline using an analytic rank-2 tangent-frame loop on a sphere. The fixed polar angle was $\theta_0=0.7$, with $\cos\theta_0=0.76484$. The exact connection is
\begin{equation}
  A_\phi
  =
  \begin{pmatrix}
  0 & -\cos\theta_0 \\
  \cos\theta_0 & 0
  \end{pmatrix},
\end{equation}
and the exact reference holonomy is $U_{\mathrm{ref}}=\exp(-2\pi A_\phi)$. The exact reference eigenphases were approximately $\{-1.47754,\;1.47754\}$. The reference unitarity error was $5.89\times10^{-17}$.

The frame-pipeline Frobenius error decreased from $9.32\times 10^{-2}$ to $5.66\times 10^{-6}$. A log--log fit gave an estimated convergence order
\begin{equation}
  2.00065.
\end{equation}
The minimum overlap singular value over the run was $\mu_{\min}=0.92074$, so the test remained well-conditioned.

This benchmark is closer to the intended overlap-based use case than the direct connection benchmark. The algorithm is not given the connection $A_\phi$ directly. It receives sampled frames, constructs the overlaps $M_k=\Phi_k^\dagger\Phi_{k+1}$, extracts the polar comparators, and forms the ordered product. Agreement with the analytic tangent-frame reference therefore validates the complete discrete reconstruction pipeline.

The minimum singular value reported in this test also confirms that convergence is not being obtained from an ill-conditioned sequence of overlaps. The adjacent subspaces remain sufficiently close throughout the refinement, so the polar comparators are stable representatives of local transport.

\subsection{Feed-forward correction}

The feed-forward validation uses the non-Abelian connection in \cref{eq:numerical_nonabelian_connection}. The ODE reference holonomy was projected to the nearest unitary by polar decomposition before constructing synthetic effective gates. The unitarity error before this projection was $5.35\times10^{-8}$, while the projected unitary reference had unitarity error $3.93\times10^{-17}$.

For both left- and right-acting conventions, the correction error converged with the same order as the holonomy estimate:
\begin{align}
  \text{holonomy order} & = 2.00855, \\
  \text{left correction order} & = 2.00855, \\
  \text{right correction order} & = 2.00855.
\end{align}
At the finest tested refinement, both left and right corrections achieved infidelity of order $1-F \simeq 3.9\times 10^{-13}$.

The correction test shows that the feed-forward error tracks the holonomy estimation error. This is precisely what the correction bound predicts: once non-geometric residual errors are absent from the synthetic model, the remaining gate error is controlled by $\|\widehat U_\gamma-U_\gamma\|$. The left- and right-acting conventions give the same convergence order in this benchmark because both are corrected using the same holonomy estimate, but they correspond to different operator orderings and should not be interchanged in a general circuit model.

\subsection{Conditioning and noise sensitivity}

The noise benchmark perturbs overlap matrices by random complex matrices with controlled norm. For the baseline frame loop with $N=80$, the minimum singular value was $\mu_{\min}=0.99872$, and the holonomy unitarity residual was $8.02\times10^{-15}$.

To test the stability scaling, each clean overlap was replaced by
\begin{equation}
  \polar(M_k)\diag(1,\mu),
\end{equation}
so that the polar factor remained unchanged while the smallest singular value was controlled. For scaled noise $\eta=\rho\mu_{\min}$, the fitted slopes of mean holonomy error versus $\rho=\eta/\mu_{\min}$ were
\begin{equation}
  \{0.995,\;1.017,\;0.947,\;1.066,\;0.983\},
\end{equation}
with mean $1.00159$. This confirms approximately linear scaling with $\eta/\mu_{\min}$ in the perturbative regime, consistent with \cref{prop:polar_stability}. In a fixed-$\eta$ test with $\eta=10^{-6}$, the fitted slope of error versus inverse conditioning was $0.36445$, showing increased sensitivity as $\mu_{\min}$ decreases.

The important feature is not the precise fitted exponent in every finite-trial noise model, but the dependence on the dimensionless perturbation ratio $\eta/\mu_{\min}$. This is the ratio predicted by the polar stability bound. The scaled-noise test confirms that, in the perturbative regime, the mean holonomy error is approximately linear in this quantity. The fixed-$\eta$ test then illustrates the complementary effect: for the same absolute perturbation, worse conditioning produces larger reconstruction error.

\begin{table*}[t]
\centering
\caption{Summary of numerical validation metrics. All data are synthetic numerical validation outputs, not experimental measurements. Numerical precision is rounded for readability.}
\label{tab:validation_summary}
\begin{tabular}{lll}
\toprule
Validation & Metric & Value \\
\midrule
Gauge covariance & Maximum covariance residual & $9.5\times10^{-15}$ \\
Gauge covariance & Maximum unitarity residual & $1.4\times10^{-14}$ \\
Non-Abelian convergence & Estimated order & $2.008$ \\
Non-Abelian convergence & Error range & $1.57\times10^{-2}\to9.01\times10^{-7}$ \\
Frame-overlap pipeline & Estimated order & $2.001$ \\
Frame-overlap pipeline & Error range & $9.32\times10^{-2}\to5.66\times10^{-6}$ \\
Frame-overlap pipeline & Minimum overlap singular value & $0.921$ \\
Feed-forward correction & Holonomy convergence order & $2.009$ \\
Feed-forward correction & Left correction order & $2.009$ \\
Feed-forward correction & Right correction order & $2.009$ \\
Feed-forward correction & Final infidelity & $\sim 3.9\times10^{-13}$ \\
Noise sensitivity & Mean scaled-noise slope & $1.002$ \\
Noise sensitivity & Fixed-$\eta$ conditioning slope & $0.364$ \\
\bottomrule
\end{tabular}
\end{table*}

\begin{figure}[t]
  \centering
  \textbf{Summary of convergence tests}\par\vspace{0.6ex}
  \includegraphics[width=\columnwidth]{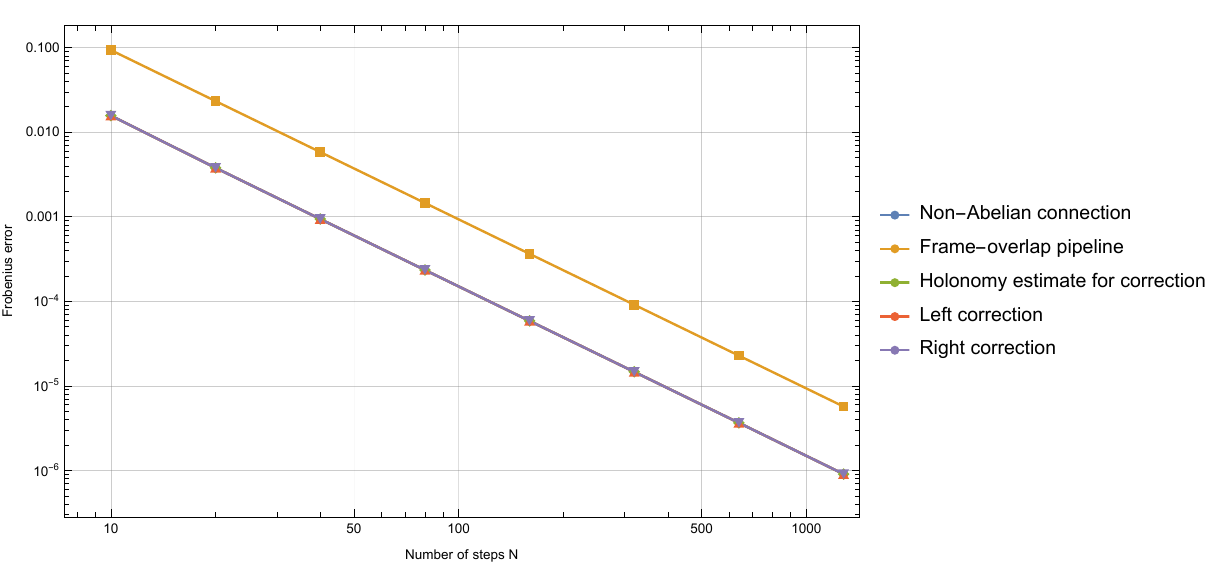}
  \caption{Aggregate convergence validation. The plot shows Frobenius errors for the non-Abelian connection benchmark, frame-overlap pipeline, holonomy estimation for correction, and left/right feed-forward correction.}
  \label{fig:summary_convergence}
\end{figure}

\begin{figure}[t]
  \centering
  \textbf{Summary of unitarity residuals}\par\vspace{0.6ex}
  \includegraphics[width=\columnwidth]{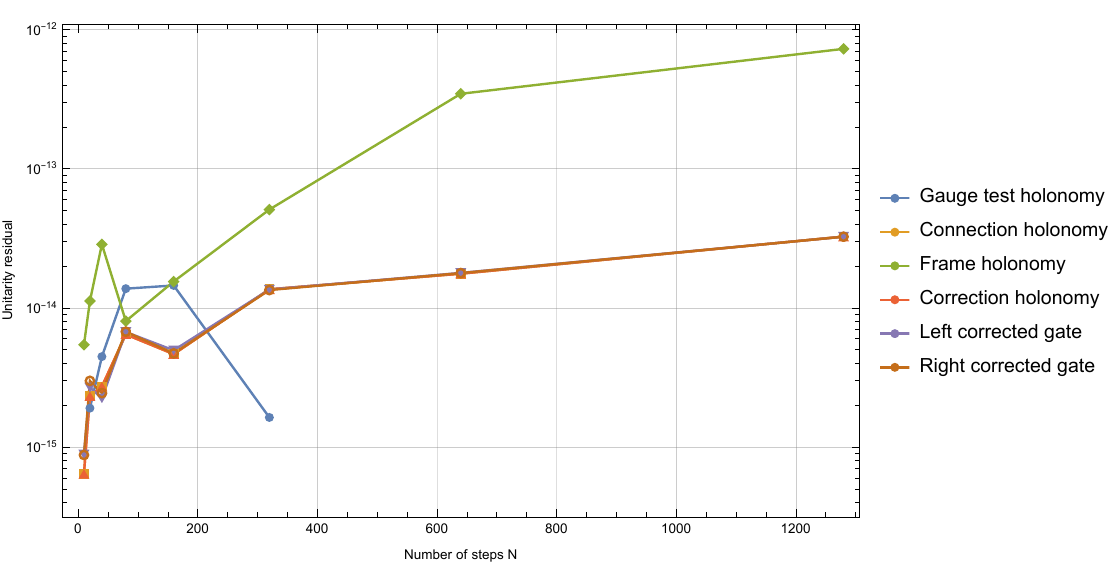}
  \caption{Aggregate unitarity validation. The reconstructed holonomies and corrected gates remain unitary to numerical precision across the tested refinements.}
  \label{fig:summary_unitarity}
\end{figure}

\begin{figure}[t]
  \centering
  \textbf{Feed-forward correction infidelity}\par\vspace{0.6ex}
  \includegraphics[width=\columnwidth]{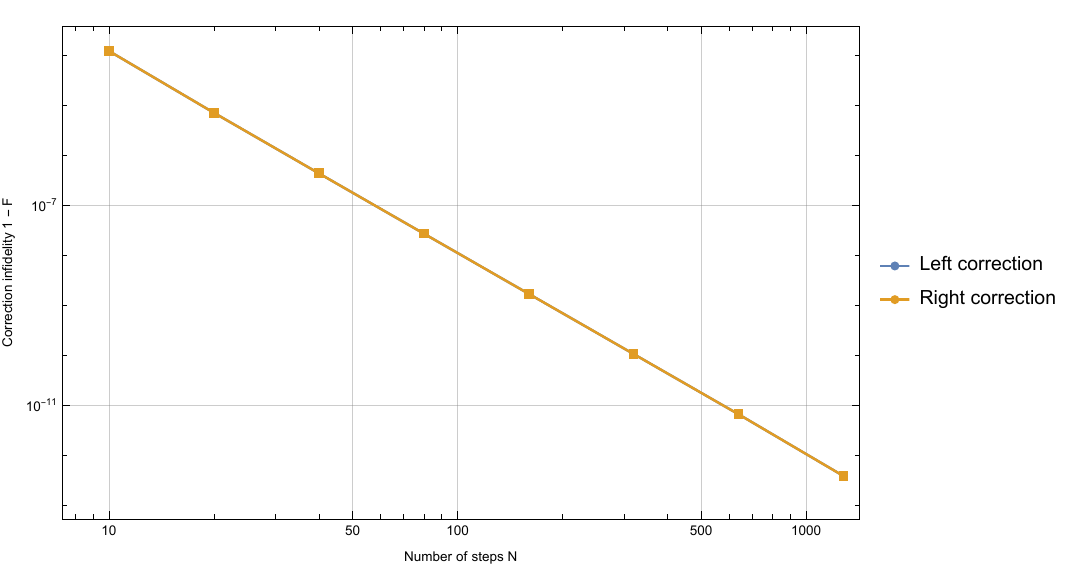}
  \caption{Feed-forward correction infidelity $1-F$. Plotting infidelity rather than fidelity resolves convergence near unit fidelity.}
  \label{fig:summary_correction_infidelity}
\end{figure}

\begin{figure}[t]
  \centering
  \textbf{Scaled-noise slopes across conditioning levels}\par\vspace{0.6ex}
  \includegraphics[width=\columnwidth]{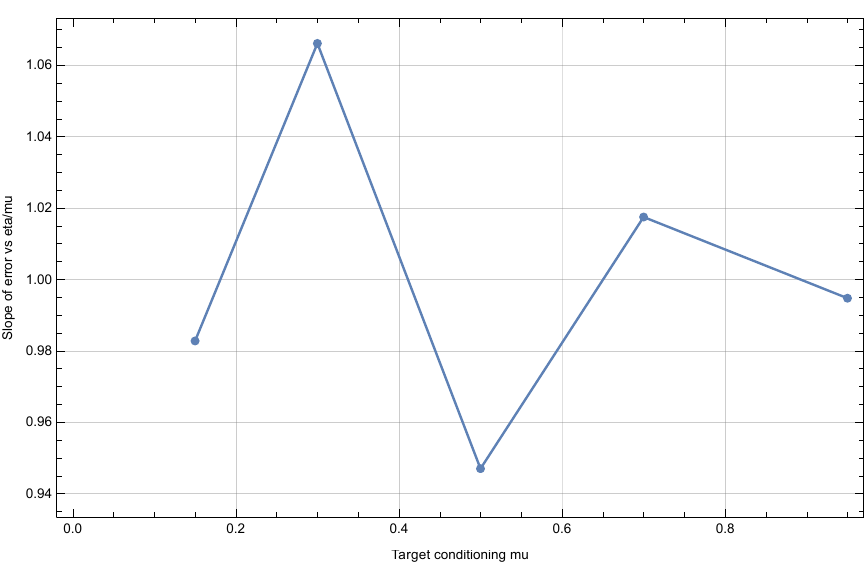}
  \caption{Scaled-noise slopes across controlled conditioning levels. The slopes cluster near one, consistent with approximately linear dependence on $\eta/\mu_{\min}$ in the perturbative regime.}
  \label{fig:summary_noise}
\end{figure}

\begin{figure}[t]
  \centering
  \textbf{Projector-step refinement diagnostic}\par\vspace{0.6ex}
  \includegraphics[width=\columnwidth]{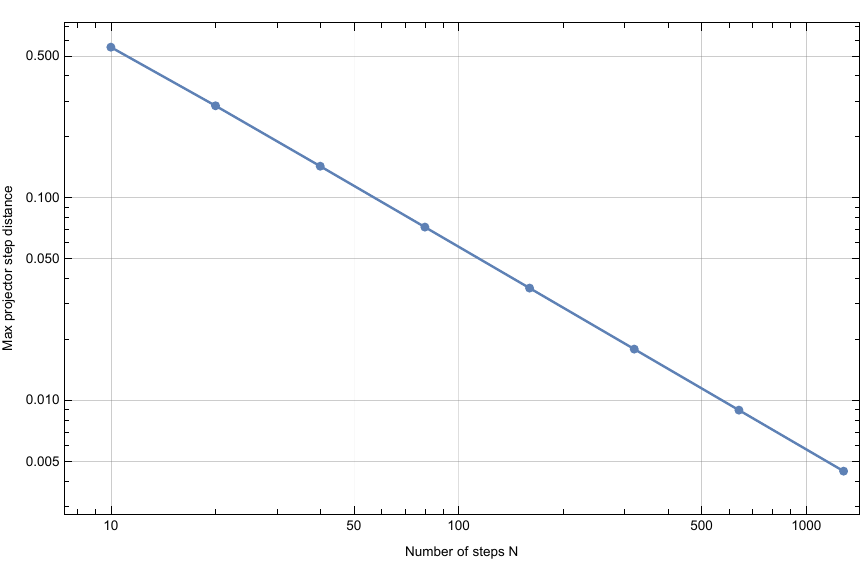}
  \caption{Projector-step refinement diagnostic for the frame-overlap pipeline. The maximum adjacent-projector distance decreases under mesh refinement, consistent with the requirement that nearby overlap matrices remain well-conditioned.}
  \label{fig:summary_projector}
\end{figure}

\section{Discussion}

The central claim of this paper is modest but structurally important: once the relevant transported object is a logical subspace rather than a collection of independent rays, Abelian phase calibration must be replaced by a gauge-covariant matrix-valued reconstruction problem. In the raywise setting, geometric distortion may be represented by scalar Pancharatnam--Berry phases and corrected by subtracting a phase profile. In the subspace setting, the corresponding geometric object is a Wilczek--Zee holonomy: a unitary matrix defined only after a choice of logical frame and transforming by conjugation under changes of that frame. In this precise sense, the present work is the non-Abelian theoretical counterpart of the Abelian time-bin calibration and feed-forward framework of Ref.~\cite{WeeBruzzese2026PB}.

The distinction between Abelian and non-Abelian geometric effects is not merely formal. In the Abelian case, geometric information is captured by a phase associated with a ray~\cite{Pancharatnam1956,Berry1984,AharonovAnandan1987,SamuelBhandari1988}. In the non-Abelian case, the transported object is an entire subspace, and the geometric information is captured by a unitary matrix defined up to conjugation~\cite{WilczekZee1984,Anandan1988}. This is the same structural feature that underlies holonomic quantum computation proposals, where logical operations are implemented or interpreted through geometric transport of subspaces~\cite{Sjoqvist2012,Xu2012PRL}. The present paper does not claim to implement holonomic computation or an experimental calibration routine. Rather, it identifies the form that a covariant calibration object must take when a time-bin logical sector is modelled as a transported subspace.

This gives the framework its main conceptual arc. The hidden simplifying assumption in an Abelian feed-forward model is that the relevant geometric distortion can be represented by commuting scalar data. Once that assumption is relaxed, three changes occur simultaneously. First, the estimated object becomes a unitary matrix rather than a phase. Second, the matrix depends on an arbitrary choice of logical frame and therefore must be handled covariantly rather than entrywise. Third, feed-forward correction becomes side-dependent: applying the inverse holonomy on the left is not equivalent to applying it on the right unless the relevant matrices commute. Thus ``subtract the phase'' is replaced by ``reconstruct the holonomy and invert it on the correct side.''

The overlap matrices $M_k=\Phi_k^\dagger\Phi_{k+1}$ contain precisely the frame-comparison data needed to reconstruct a discrete version of Wilczek--Zee transport. The polar decomposition extracts a unitary comparator from each overlap, ensuring that the reconstructed evolution remains on $U(m)$ even though the raw overlaps need not be unitary. Moreover, by polar optimality, this comparator is the nearest unitary matrix to the raw overlap in Frobenius norm. The method therefore gives a canonical local unitarisation of frame-comparison data.

Another way to view the construction is as a discrete parallel-transport rule. The raw overlap matrix answers the question: how does the next frame project onto the current one? The polar factor extracts from that answer the unitary part of the comparison, while the singular values record whether the comparison was trustworthy. This separation is useful because it keeps the reconstructed transport on the unitary group without hiding the conditioning information.

The gauge structure is not a technical complication but a necessary part of the formulation. Since the choice of frame inside the logical subspace is arbitrary, the reconstructed closed-loop holonomy can only be expected to transform by conjugation. Consequently, physical reporting should use conjugacy invariants, such as eigenphases or Wilson-loop traces.

This point is particularly important for reporting numerical or experimental results. A different choice of logical basis can change every entry of the holonomy matrix, even though the underlying closed-loop transport is the same. For this reason, tables of raw matrix entries are less meaningful than eigenphases, traces of powers, and conditioning diagnostics. The validation suite follows this principle by reporting gauge-covariant residuals and gauge-invariant spectral data.

The feed-forward correction principle generalises Abelian geometric phase cancellation. In the scalar case, one estimates a phase and subtracts it. In the non-Abelian case, one estimates a unitary matrix and applies its inverse on the appropriate side of the effective logical operation. The side matters:
\begin{equation}
  V_{\mathrm{eff}}\approx U_\gamma V
  \quad\Rightarrow\quad
  V_{\mathrm{corr}}=\widehat U_\gamma^\dagger V_{\mathrm{eff}},
\end{equation}
whereas
\begin{equation}
  V_{\mathrm{eff}}\approx VU_\gamma
  \quad\Rightarrow\quad
  V_{\mathrm{corr}}=V_{\mathrm{eff}}\widehat U_\gamma^\dagger.
\end{equation}
These two correction rules coincide only in special commuting cases. In any device-specific application, the correct side must be fixed by the physical circuit convention, coefficient convention, or effective model used to identify the geometric factor.

The theory also identifies a practical limitation. The method is stable only when adjacent overlap matrices are well-conditioned. Thus a refined discretisation is not merely a numerical convenience; it is part of the validity condition for reliable holonomy reconstruction. The smallest singular value diagnostic $\mu_{\min}$ should therefore be reported in any theoretical or numerical validation.

The numerical validation clarifies the role of the theory. The tests do not constitute an experimental demonstration; they are reproducible synthetic benchmarks of the mathematical framework. Their value is that they verify the algorithmic properties that any later device-specific implementation would need: gauge covariance, correct path ordering, stable overlap reconstruction, and successful feed-forward correction when the geometric contribution is known or estimated.

The present work is intentionally platform-independent. It does not assume a specific integrated photonic circuit, loss model, detector model, or experimental calibration pipeline. Device-specific implementations can be built on top of this framework by supplying a physical model for the frames or overlaps. In that sense, the framework is best read not as a completed experimental protocol, but as a theoretical account of what calibration must track once the relevant time-bin logical object is treated as a transported subspace rather than as a collection of independent rays.

\section{Code and data availability}

The accompanying repository contains Mathematica/Wolfram Language notebooks and a headless validation pipeline that reproduce the numerical validation tables, figures, and logs used in this manuscript. The notebooks are organised as independent validation modules: gauge covariance, non-Abelian convergence, frame-overlap reconstruction, feed-forward correction, conditioning/noise sensitivity, and aggregate summary. The repository also includes shared source routines for polar holonomy construction, overlap processing, unitary diagnostics, random gauge transformations, and export of reproducibility artefacts. The generated data are synthetic numerical validation outputs rather than experimental measurements.

The reproducibility package is archived on Zenodo with DOI~\href{\zenodourl}{\zenododoi}; the source is on \href{\repositoryurl}{GitHub}.

\section{Conclusion}

We developed a gauge-covariant framework for reconstructing non-Abelian holonomies from discrete overlap matrices between transported logical subspace frames. The polar factor of each overlap matrix defines a canonical backward comparator, whose adjoint gives the forward coefficient transport. The resulting ordered product is gauge-covariant, converges to the Wilczek--Zee holonomy under partition refinement, and is stable when overlap matrices remain well-conditioned.

The framework generalises Abelian time-bin geometric-phase correction to the non-Abelian setting. Scalar phase subtraction is replaced by matrix-valued feed-forward correction, with distinct left- and right-acting forms depending on the circuit convention. Reproducible numerical validation confirms gauge covariance, convergence, feed-forward correction fidelity, and the predicted dependence on overlap conditioning.

\appendix
\section{Proof details}

This appendix collects several technical details used in the main text. The purpose is to make explicit the sign convention, the polar-decomposition covariance, and the relationship between frame comparison and coefficient transport.

\subsection{Equivariance of the polar decomposition}

Let
\begin{equation}
  M=WH
\end{equation}
be the polar decomposition of a nonsingular matrix $M\in\mathbb C^{m\times m}$, with $W\in U(m)$ and $H=(M^\dagger M)^{1/2}>0$. Let $L,R\in U(m)$. Then
\begin{equation}
  LMR
  =
  (LWR)(R^\dagger H R).
\end{equation}
The first factor $LWR$ is unitary, and the second factor $R^\dagger H R$ is positive Hermitian. By uniqueness of the polar decomposition for nonsingular matrices,
\begin{equation}
  \polar(LMR)=L\polar(M)R.
  \label{eq:appendix_polar_equivariance}
\end{equation}
This identity is the key algebraic input in the gauge-covariance proof.

\subsection{First-order expansion of the overlap matrix}

Let $\Phi(t)$ be a $C^2$ orthonormal frame and let
\begin{equation}
  A_t(t)=\Phi(t)^\dagger\dot\Phi(t).
\end{equation}
For a small step $\Delta t$, Taylor expansion gives
\begin{equation}
  \Phi(t+\Delta t)
  =
  \Phi(t)+\dot\Phi(t)\Delta t+O(\Delta t^2).
\end{equation}
Hence the adjacent-frame overlap is
\begin{align}
  M(t,\Delta t)
  &:=
  \Phi(t)^\dagger\Phi(t+\Delta t) \\
  &=
  I_m+\Phi(t)^\dagger\dot\Phi(t)\Delta t+O(\Delta t^2) \\
  &=
  I_m+A_t(t)\Delta t+O(\Delta t^2).
  \label{eq:appendix_overlap_expansion}
\end{align}
Since $\Phi(t)^\dagger\Phi(t)=I_m$, differentiation gives $A_t(t)^\dagger=-A_t(t)$.

\subsection{First-order expansion of the polar comparator}

Using \cref{eq:appendix_overlap_expansion}, write
\begin{equation}
  M=I_m+A\Delta t+O(\Delta t^2),
\end{equation}
with $A^\dagger=-A$. Then
\begin{equation}
  M^\dagger M
  =
  I_m+O(\Delta t^2).
\end{equation}
Therefore
\begin{equation}
  (M^\dagger M)^{-1/2}
  =
  I_m+O(\Delta t^2),
\end{equation}
and the polar unitary satisfies
\begin{align}
  \polar(M)
  &=
  M(M^\dagger M)^{-1/2} \\
  &=
  I_m+A\Delta t+O(\Delta t^2).
  \label{eq:appendix_polar_expansion}
\end{align}
Thus the polar factor of $M_k=\Phi_k^\dagger\Phi_{k+1}$ has the same first-order sign as the connection:
\begin{equation}
  W_k
  =
  I_m+A_k\Delta t_k+O(\Delta t_k^2).
\end{equation}

\subsection{Discrete-to-continuum sign convention}

Let
\begin{equation}
  |\psi(t)\rangle=\Phi(t)c(t)
\end{equation}
be a state vector lying in the transported logical subspace. Projector parallel transport is characterised by
\begin{equation}
  P(t)|\psi(t)\rangle=|\psi(t)\rangle,
  \qquad
  P(t)|\dot\psi(t)\rangle=0,
  \label{eq:appendix_projector_parallel_condition}
\end{equation}
where $P(t)=\Phi(t)\Phi(t)^\dagger$. Using
\begin{equation}
  |\dot\psi(t)\rangle=\dot\Phi(t)c(t)+\Phi(t)\dot c(t),
\end{equation}
we obtain
\begin{align}
  0
  &=
  P|\dot\psi\rangle \\
  &=
  \Phi\Phi^\dagger(\dot\Phi c+\Phi\dot c) \\
  &=
  \Phi(A_t c+\dot c).
\end{align}
Since $\Phi$ has full column rank,
\begin{equation}
  \dot c(t)=-A_t(t)c(t).
  \label{eq:appendix_coefficient_transport}
\end{equation}
Thus the forward coefficient transport over a small step is
\begin{equation}
  c(t+\Delta t)
  =
  \left[I_m-A_t(t)\Delta t+O(\Delta t^2)\right]c(t).
\end{equation}
On the other hand, the polar unitary extracted from the overlap matrix satisfies
\begin{equation}
  W_k
  =
  I_m+A_k\Delta t_k+O(\Delta t_k^2).
\end{equation}
Therefore $W_k$ is naturally a backward frame comparator. The forward coefficient transport is its adjoint:
\begin{equation}
  T_k=W_k^\dagger
  =
  I_m-A_k\Delta t_k+O(\Delta t_k^2).
  \label{eq:appendix_forward_step_expansion}
\end{equation}

\subsection{Product ordering}

The convention in this paper uses column coefficient vectors. If $c_{k+1}=T_kc_k$, then iterating gives
\begin{equation}
  c_N
  =
  T_{N-1}T_{N-2}\cdots T_1T_0c_0.
\end{equation}
Thus the discrete holonomy acting on column vectors is
\begin{equation}
  \widehat U_\gamma
  =
  T_{N-1}T_{N-2}\cdots T_1T_0.
\end{equation}

\subsection{Telescoping identity for product perturbations}

Let
\begin{align}
  U &= T_{N-1}\cdots T_0, \\
  \widetilde U &= \widetilde T_{N-1}\cdots\widetilde T_0.
\end{align}
Then
\begin{multline}
  \widetilde U-U
  \\
  =
  \sum_{j=0}^{N-1}
  \widetilde T_{N-1}\cdots\widetilde T_{j+1}
  (\widetilde T_j-T_j)
  T_{j-1}\cdots T_0.
  \label{eq:appendix_telescoping_identity}
\end{multline}
If all $T_k$ and $\widetilde T_k$ are unitary, then taking the operator norm gives
\begin{equation}
  \norm{\widetilde U-U}_2
  \le
  \sum_{j=0}^{N-1}
  \norm{\widetilde T_j-T_j}_2.
\end{equation}
Since $T_j=W_j^\dagger$ and $\widetilde T_j=\widetilde W_j^\dagger$, one also has
\begin{equation}
  \norm{\widetilde T_j-T_j}_2
  =
  \norm{\widetilde W_j-W_j}_2.
\end{equation}

\bibliographystyle{unsrt}
\bibliography{references}

\end{document}